\documentclass[a4paper,reqno]{amsart}
\usepackage[utf8]{inputenc}
\usepackage[T1]{fontenc}

\usepackage{amsthm,amsmath,amsfonts,amssymb,dsfont,bm,amstext,amsopn,mathrsfs,mathtools,esint}
\usepackage[colorlinks, linkcolor={blue}, citecolor={red}]{hyperref}
\usepackage{float}
\numberwithin{equation}{section} % For numbering the equations with n°section.
\newtheorem{theorem}{Theorem}[section]
\newtheorem{proposition}[theorem]{Proposition}

\newtheorem{lemma}[theorem]{Lemma}

\newcommand\C{{\mathbb C}}
\newcommand\LL{{\mathbb L}}

\newcommand\R{{\mathbb R}}
\newcommand\Z{{\mathbb Z}}

\newcommand\ba{{\mathbf a}}
\newcommand\bb{{\mathbf b}}
\newcommand\bh{{\mathbf h}}

\newcommand\bk{{\mathbf k}}
\newcommand\bt{{\mathbf t}}
\newcommand{\bd}{\mathbf{d}}
\newcommand{\bK}{\mathbf{K}}

\newcommand\bnull{{\mathbf 0}}
\newcommand\sC{{\mathscr{C}}}

\newcommand\cE{{\mathcal E}}
\newcommand{\mce}{\mathcal{E}}
\newcommand\cF{{\mathcal F}}

\newcommand\rd{{\mathrm{d}}} % for \rd x in the integral
\newcommand\re{{\mathrm{e}}} % the exponential map
\newcommand\ri{{\mathrm{i}}} % the number \ri so that \ri^2 = -1.

\renewcommand\Re{\rm Re}
\newcommand\VTr{{\underline{\rm Tr}}}

\newcommand{\dee}{\:\mathrm{d}}
\newcommand{\dist}{{\rm dist}}

\DeclareMathOperator{\Tr}{Tr}

\DeclareMathOperator{\e}{e}

\newcommand{\ps}[2]{\langle #1,#2\rangle}

\usepackage[backend=biber,
style=alphabetic,
maxbibnames=99,
doi=false, url=false, giveninits=true, date=year]{biblatex}
\usepackage{csquotes}

\title[Peierls Instability in hexagonal lattices]{Peierls Instability in hexagonal Lattices}

\author{David Gontier}
\address[David Gontier]{CEREMADE, University of Paris-Dauphine, PSL University, 75016 Paris, France.}
\email{gontier@ceremade.dauphine.fr}

\author{Thaddeus Roussigné}
\address[Thaddeus Roussigné]{CEREMADE, University of Paris-Dauphine, PSL University, 75016 Paris, France.}
\email{roussigne@ceremade.dauphine.fr}

\author{Éric Séré}
\address[Éric Séré]{CEREMADE, University of Paris-Dauphine, PSL University, 75016 Paris, France.}
\email{sere@ceremade.dauphine.fr}

\date{27th October 2025}

\begin{document}

\maketitle

\begin{abstract}
    We investigate a conventional tight-binding model for graphene, where distortion of the honeycomb lattice is allowed, but penalized by a quadratic energy. We prove that the optimal 3-periodic lattice configuration has Kekulé O-type symmetry, and that for a sufficiently small elasticity parameter, the minimizer is not translation-invariant. Conversely, we prove that for a large elasticity parameter the translation-invariant configuration is the unique minimizer.  
\end{abstract}

%%%%%%%%%%%%%%%%%%%%%%%%%
\section{Introduction}

When studying closed one-dimensional atomic chains, such as polyacetylene, one encounters a distortion by which the ground state electronic configuration \textit{dimerises}. Despite the energy functional being translation-invariant, its minimizers turn out to be 2-periodic, rather than 1-periodic with respect to the chain. This phenomenon was first encountered and explained in the 1950s by Peierls \cite{peierls1955quantum}, Fröhlich \cite{frohlich1954theory}, Longuet-Higgins and Salet \cite{longuet1959alternation}, and is commonly referred to as the Peierls instability or Peierls distortion~\cite{kennedy1987proof, lieb1995stability, gontier2023phase}. 

In this paper, we investigate an analogous distortion for two-dimensional honeycomb crystals such as graphene. In the thermodynamic limit, the possible distortions of graphene-like lattices for the Hubbard energy were investigated by Frank and Lieb in \cite{frank2011possible}. They showed that, although the model is $1$--periodic with respect to the lattice vectors, the ground state energy is attained for a configuration which is at most $3$--periodic, and furthermore has at most three different hopping parameters by way of the lattice's symmetries (see below for details). They also conjectured that all minimizers obey this structure. We build on these results in the Hückel (non-interacting) case with a quadratic elastic energy, and prove that among the set of 3-periodic and honeycomb-symmetric configurations, the minimizers have the additional \textit{Kekulé O-type} symmetry. This means that two of the three aforementioned hopping values are equal. We also prove that when the lattice rigidity is small, the energy inequality is strict and the configuration is truly distorted, and as the rigidity increases, the three parameters become equal and we recover pristine graphene. 

In the strict Kekulé distortion case, we can roughly interpret the bonds corresponding to the two equal hopping parameters as longer single covalent bonds, and the remaining hopping value as a shorter double bond, recovering the carbon atom's valence of four.  The overall Kekulé texture is reminiscent of the celebrated resonance forms for benzene, first posited by Kekulé in the 1860s \cite{kekule1865constitution}, hence the "O-type" designation, itself initially applied to graphene in 2000 by Chamon in \cite{chamon2000solitons}. Such a distorted material is an insulator (a well-known fact, that we prove below in Lemma~\ref{lem:spectrum_T}). But when all three parameters are equal, the corresponding model is metallic. Our result therefore shows that there is phase transition from an insulating Kekulé phase to the metallic phase (pristine graphene) as the lattice rigidity increases.

%%%%%%%%%%%%%%%
\subsection{Peierls tight-binding energy on the hexagonal lattice}

We consider a tight-binding Hückel model (sometimes called Peierls model or Su--Schrieffer--Heeger model~\cite{su1979solitons}) for the entire graphene plane, which we recall has a hexagonal structure with a 2--atom unit cell. In this model, graphene is parametrised by the real-valued hopping parameter vector $\bt = \{ t_{xy}\}$ and the distance vector $\bd = \{ d_{xy}\}$ over pairs of atoms $x,y$. We now assume that $\bt$ and $\bd$ are $L$--periodic with respect to the $2$--atom unit cell for some integer $L$, meaning it suffices to know the values they take within an $L \times L$ supercell of unit cells to describe the whole lattice. Then the {\em energy per Carbon atom} writes as follows:
\[
    \cE(\bt, \bd) :=  - \VTr \left| T(  \bt ) \right| + \frac{\mu}{2} \left( \frac{1}{2 L^2} \right) \sum_{\ps{x}{y}} (d_{xy} - d_\sharp)^2.
\]
The first term describes the quantum energy of graphene's $\pi$-electrons under the tight-binding Hamiltonian $T ( \bt)$, which has coefficients $t_{xy}$ if the Carbon atoms labelled $x$ and $y$ are nearest neighbours for the hexagonal lattice, and $0$ otherwise. We denote by $\VTr$ the trace per unit Carbon atom, which is well-defined for operators $T(\bt)$ that commute with $L$--translations, as is the case under our previous assumption. We recall its definition below in Section~\ref{section:first_facts}.
%This particular expression with the absolute value comes from the minimization of $2 \Tr_\Lambda (T(\widetilde{\bt}) \gamma)$ over all one-body density matrices $\gamma$ satisfying the Pauli principle $0 \le \gamma = \gamma^* \le 1$. The $2$ factor stands for the spin of the electrons.
%
The second term describes a classical elastic interaction between Carbon nuclei, modelled by a quadratic distortion energy. The sum runs only over nearest neighbours, $d_{xy}$ is the distance between Carbon $x$ and $y$, and $d_\sharp$ is the resting distance. The $\frac{1}{2L^2}$ factor comes from the fact that there are $2 L^2$ Carbon atoms in a given $L \times L$ supercell. Following \cite{kennedy1987proof}, assume a linear relation between $t_{xy}$ and $d_{xy}$, of the form
\[
    t_{xy} - t_\sharp = - \alpha (d_{xy} - d_\sharp),
\]
for some $\alpha > 0$, and after the change of variables $\widetilde{\bt} = \bt/t_\sharp$, $\widetilde{\mu} = \mu t_\sharp^2/\alpha^2$, and $\widetilde{\cE} = \cE/t_\sharp$, we obtain the equivalent energy (dropping the tilde notations)
\begin{equation*}
     \cE(\bt) := - \VTr \left| T( \bt ) \right| + \frac{\mu }{2} \left( \frac{1}{2L^2} \right) \sum_{\ps{x}{y}} (t_{xy} - 1)^2 ,
\end{equation*}
with a single parameter left in the model, namely $\mu > 0$, which we call the {\em lattice rigidity}, or elasticity parameter. This model is the $2$--dimensional analogue of the usual Hückel model for polyacetylene. It describes a competition between the distortion energy of the molecule, and the quantum energy of the $\pi$--electrons. 

\medskip

The study of the energy $\cE$ turns out to be quite difficult. Actually, should we consider minimizers among periodic $\bt$, with different periods $L_1$ and $L_2$ along either lattice direction, the situation may be very different. 

Nonetheless, in~\cite{frank2011possible}, using a reflection positivity argument (see also \cite{lieb1995stability}), Frank and Lieb proved that, in the thermodynamic limit $L_1, L_2 \to \infty$, the infinite graphene sheet has a minimizer which is at most $3$-periodic, and with at most $3$ different values for the hopping parameters, arranged according to the lattice's axial symmetries. 

To be more specific, let $\ba_1, \ba_2 := (\sqrt{3}/2, \pm1/2)$ and $\LL_2 := \ba_1 \Z \oplus \ba_2 \Z$ be the usual triangular lattice. The $2$ subscript refers to the fact that graphene has $2$ Carbon atoms per unit cell of this lattice. We denote this unit cell by $\Gamma_2 := \R^2 / \LL_2$. We now set $\bb_1 := 2 \ba_1 - \ba_2$ and $\bb_2 := 2 \ba_2 - \ba_1$, and $\LL_6 := \bb_1 \Z + \bb_2 \Z$ a triangular sublattice of $\LL_2$, so that graphene has $6$ Carbon atoms per unit cells of $\LL_6$, which we denote by $\Gamma_6 := \R^2 / \LL_6$. See also Figure~\ref{fig:honeycomb1}. It is proved in~\cite{frank2011possible} that there is a minimizer $\bt$ of infinite graphene which is $\LL_6$--periodic. In addition, this configuration only has {\em three} different hopping amplitudes, which we denote by $t,u,v$, as displayed in Figure~\ref{fig:honeycomb1}, satisfying all axial symmetries with respect to axes going through covalent bonds. It is also conjectured in \cite{frank2011possible} that in fact, all minimizers have this 3-periodic symmetric structure.

\medskip

\begin{figure}[!h]
    \centering
    \includegraphics[width = 0.35\linewidth]{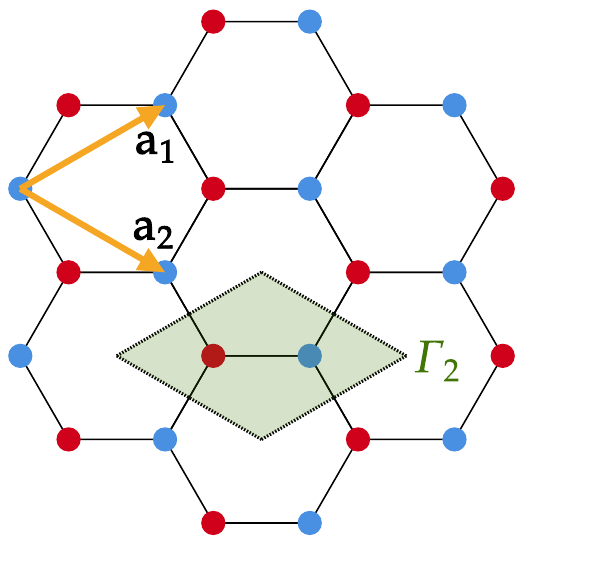}
    \includegraphics[width = 0.35\linewidth]{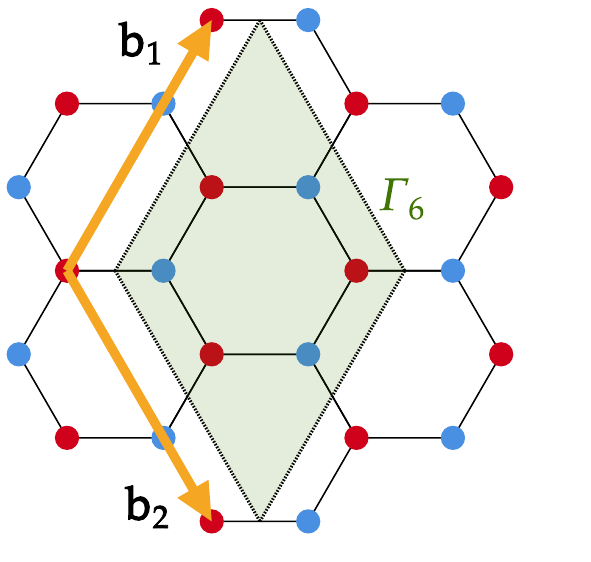} 
    \includegraphics[width = 0.5\linewidth]{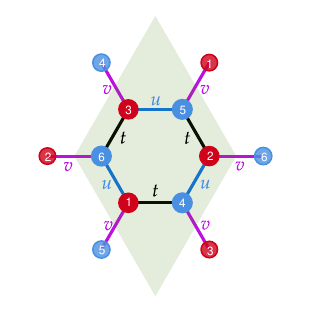}
    \caption{(Above, left) The standard honeycomb lattice, with unit vectors $\ba_1, \ba_2$ and the unit cell $\Gamma_2$. The blue and red atoms make up its two triangular sublattices. (Above, right) The larger unit cell $\Gamma_6$ defined with the new vectors $\bb_1, \bb_2$, containing six atoms and nine bonds. (Below) A minimizing configuration following \cite{frank2011possible}, with our conventions for the atom and bond labels.}
    \label{fig:honeycomb1}
\end{figure}

\medskip

In the sequel we take this as a framework and investigate minimizers only among 3-periodic configurations, and having the symmetries described in~\cite{frank2011possible}. For such infinite configurations, the energy per unit Carbon atom is well defined, depends solely on the three remaining hopping parameters $t,u,v$, and is given by
\[
    \boxed{\mce(t,u,v) = - \VTr \left| T(t,u,v) \right| + \frac{\mu}{2}\cdot \frac12 \left((t-1)^2 + (u-1)^2 + (v-1)^2\right).}
\]
We used here that the unit cell $\Gamma_6$ contains $6$ carbon atoms and $3$ bonds of each of the $t,u,v$ amplitudes, thus each bond has a $1/2$ coefficient. Note that $\mce$ is invariant under permutations of $t,u,v$, since they play symmetric roles in the distortion (alternatively, a translation by the $\ba_1$--vector acts as the permutation $(t,u,v) \to (v, t, u)$). The precise form of the hopping operator $T$ follows from our labelling, see Figure~\ref{fig:honeycomb1} and Section~\ref{section:first_facts}.

\medskip

\subsection{Main results}
Our goal is to study the minimizers of $\mce$. If the minimizer is such that two of the parameters are equal, we say that this minimizer has {\em Kekulé $O$--symmetry}. If in addition the three parameters are equal, we recover pristine configurations. We can now state our main results.
\begin{theorem}[Kekulé symmetry] \label{thm:Ksym_short}
    For any rigidity $\mu > 0$, the minimizers of $\mce$ are non-negative and satisfy the Kekulé $O$-type symmetry.
\end{theorem}

In other words, in the non-interacting Hubbard case with quadratic distortion energy, two of the three amplitudes prescribed in \cite{frank2011possible} are equal. The proof relies entirely on the quadratic expression of the distortion energy. We do not know how to generalise the result to more general elastic terms, nor to the interacting Hubbard case. 
 
\begin{theorem}[Phase transition]\label{thm:phasetrans}
    There are $0 < \mu_c \le \mu_c' < \infty$ such that: 
    \begin{enumerate}
        \item[(i)]
        For all $0 < \mu < \mu_c$, the minimizers of $\mce$ satisfy $t = u < v$ (up to permutation). In particular, minimizers of $\mce$ are strictly Kekulé distorted, and the corresponding tight-binding Hamiltonian $T$ is gapped, satisfying $0 \notin \sigma(T)$ (insulating case).
        \item[(ii)]
        For all $\mu \ge \mu_c'$, the minimizer of $\mce$ is unique, and of the form $t = u = v$. In this case, graphene is pristine, and $0 \in \sigma(T)$ (conducting case).
    \end{enumerate}
\end{theorem}

In the first case where $t = u < v$, the lattice is distorted as represented in Figure~\ref{fig:kekfig}. It features a periodic spread of hexagons, either regular with long $t$ bonds, or alternating long $t$ and short $v$ bonds.

Numerically, we find that $\mu_c \geq 0.888$ and $\mu_c' \le 1.114$, and we conjecture that $\mu_c = \mu_c' \approx 0.888$. We therefore expect that there is a true transition threshold, separating the Kekulé O-type and pristine phases. Experimentally, isolated graphene always exhibits ``metallic'' behaviour, including at very low temperatures \cite{Andrade_2025}, and accordingly it seems that the appropriate value of $\mu$ is above the transition value of our model. Nonetheless, the Kekulé phase remains physically relevant: experiments conducted in \cite{Eom2020direct} and simulations in \cite{otsuka2024kekule} show that pristine graphene can be driven into a Kek-O distorted configuration by application of strain, exhibiting exactly the phase transition obtained here. Furthermore the nature of this transition is under study as detailed in \cite{li2017fermion}. It can also appear as an effect of growing graphene on a substrate, for instance by being intercalated with lithium atoms \cite{Bao2021evidence}. See \cite{Andrade_2025} for a review of current topics in Kekulé-distorted graphene.  

\medskip

To prove the first point of Theorem~\ref{thm:phasetrans}, we compute the linear perturbation around the optimal pristine configuration, which prevents us from obtaining uniqueness of the distorted minimizers up to permutation. We prove the second point using a convexity argument, which actually applies to periodic configurations of any period $L$, and not just $L=3$ (inherited from the reasoning in \cite{frank2011possible}).  The second point can also be generalised to a slightly larger class of distortion functionals, see Appendix~\ref{appendix:gen_elastic}. 

\medskip

We conclude this section with a comparison with the aforementioned Peierls instability for one-dimensional chains such as polyacetylene. Consider a finite Carbon chain on the periodised integer set $[\![ 1, \cdots, 2N ]\!]$ for even $N$. Then the minimizers to the finite-chain Peierls energy must be of the form $t_i = E + (-1)^i\delta$ with $\delta > 0$, as shown in \cite{kennedy1987proof} and in \cite{lieb1995stability}. The chain is said to be \textit{dimerised}, spontaneously breaking translation invariance. This distortion persists in the thermodynamic limit (see for instance \cite{gontier2023phase} for a discussion including positive temperatures), and it remains present for any rigidity $\mu$. The Kekulé distortion discussed here is in some sense an analogue to dimerisation in two dimensions, since it is periodic with respect to an enlarged unit cell, and in both cases the kinetic energy is lowered by opening a gap in the band diagram. The absence of a phase transition in dimension 1 is a dimensional effect, by which the quantum energy term is always dominant over the distortion term when considering a perturbation around the pristine configuration.

%%%%%%%%%%%%%%%%%%%%%%%%%
\section{First facts on the Kekulé tight--binding operator}\label{section:first_facts}

Let us first spell out the operator $T(t,u,v)$, and explain our definition for the trace per Carbon atom.

\begin{figure}[ht]
    \centering
    \includegraphics[scale=0.8]{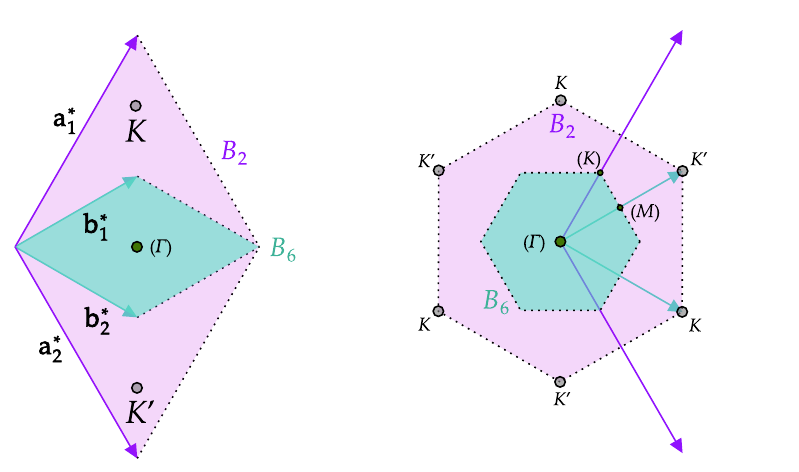}
    \\
    \caption{The Brillouin zones $B_2$ and $B_6$ are the unit cells of the reciprocal lattices $\LL_2^*$ and $\LL_6^*$, represented in rhombic (left) and Voronoi or hexagonal forms (right). Note that $B_6$ is exactly one-third of $B_2$, and the two Dirac points $\bK, \bK' \in B_2$ are ``folded'' onto central point $(\Gamma)$ when enlarging unit cell $\Gamma_2$ to $\Gamma_6$. The high-symmetry points $(K)$ and $(M)$ are used for the band diagram in Figure~\ref{fig:bdiagram} below.}
    \label{fig:brillouin}
\end{figure}

%%%%%%%%%%%%%%%%%
\subsection{The trace per Carbon atom}

Consider a lattice $\LL$, an $\LL$-periodic discrete set $\Lambda$ and an operator $H$ acting on $\Lambda$ which commutes with $\LL$--translations. We denote by $N$ the number of points of $\Lambda$ per unit cell of $\LL$. For instance, if $\Lambda$ is the honeycomb, and $\LL_2$ is the usual triangular lattice, then $N = 2$, and if $\LL_6$ is the triangular sub-lattice of the previous section, then $N = 6$.

\medskip

By periodicity, the operator $H$ acts as a multiplicative operator over $\C^N$ in Fourier space. This means that for all $f \in \ell^2(\Lambda) \approx \ell^2(\LL, \C^N)$, we have, for all $\bk \in B := \R^2 / \LL$,
\[
    \cF_\LL \left[ Hf \right](\bk) = H(\bk) \cF_\LL\left[ f\right](\bk),
\]
where $\cF_\LL$ is the Fourier transform relative to the $\LL$--lattice. Note that for all $\bk$, $H(\bk)$ is an $N \times N$ matrix. Then, we define the {\em trace per Carbon atom} of $H$ by
\[
    \VTr(H) := \frac{1}{N} \fint_{B} \Tr_{\C^N} \left[ H(\bk) \right] \rd \bk,
\]
with $\fint_B := \frac{1}{| B |} \int$. This definition is independent of the choice of the lattice $\LL$, in the following sense: if $\LL'$ is a sublattice of $\LL$, then $H$ also commutes with $\LL'$--translations. In this case, $N'$ is a multiple of $N$, $B'$ is a fraction of $B$, and so on. However the normalisation is chosen so that $\VTr'(H) = \VTr(H)$. Note for instance that for $H = I$ the identity operator, we have $\VTr(H) = 1$.

%%%%%%%%%%%%%%%%%

\subsection{The tight--binding operator}

Consider the operator $T := T(t,u,v)$, acting on the full hexagonal lattice, which commutes with $\LL_6$--translations, and with hopping parameters given as in Figure~\ref{fig:honeycomb1}. By periodicity, and since $\Gamma_6$ contains $6$ Carbon atoms, $T(t,u,v)$ acts as multiplication operator over $\C^6$ in Fourier space (with respect to the $\LL_6$--lattice). More specifically, consider the dual lattice $\LL_6^* := \bb_1^* \Z \oplus \bb_2^* \Z$ with $\bb_i \cdot \bb_j^*  = 2 \pi \delta_{ij}$, and the Brillouin zone $B_6 := \R^2 / \LL_6^*$ (see Figure~\ref{fig:brillouin}). Then, for $\bk \in B_6$, we get, with our convention,

\[
    T(\bk) = \begin{pmatrix} 0 & A(\bk) \\ A(\bk)^* & 0 \end{pmatrix}, ~\text{where }
A(\bk) := \begin{pmatrix} t & v\e^{-\ri\bk \cdot \bb_1} & u \\ u & t & v \e^{\ri(\bk \cdot \bb_1 + \bk \cdot \bb_2)} \\ v\e^{-\ri\bk \cdot \bb_2} & u & t \end{pmatrix}.
\]
To illustrate: an electron at site (4) may hop to neighbouring site (1) with amplitude $t$, to site (2) with amplitude $u$, and to site (3) in the next cell along by a translation by $\bb_2$, with amplitude $v$. This sets the fourth column in $T(\bk)$. 

\medskip

In our energy $\mce$, we are interested in the quantity $| T | := \sqrt{T^2}$. Note that 
\[
    T^2 = \begin{pmatrix} AA^* & 0 \\ 0 & A^*A \end{pmatrix}, \quad \text{hence} \quad
    | T | = \begin{pmatrix} \sqrt{AA^*} & 0 \\ 0 & \sqrt{A^*A} \end{pmatrix}.
\]
Recall that $\sigma(A A^*) = \sigma(A^* A)$ (except maybe for the eigenvalue $0$). So the trace per Carbon atom used in the energy $\mce$ is also
\[
    \VTr(| T |) := \frac16 \fint_{B_6}   \Tr_{\C^6} \left( |T|(\bk) \right) \rd \bk 
                        = \frac{1}{3} \fint_{B_6}   \Tr_{\C^3} \left( \sqrt{A^* A}(\bk) \right) \rd \bk.
\]

We record the following Lemma.
\begin{lemma} \label{lem:spectrum_T}
    Let $t,u,v \ge 0$ and consider their average and variance given by
    \[
        E := \frac13 (t + u + v), \quad \text{and} \quad s^2 := \frac{1}{3} \left[ (t-E)^2 + (u - E)^2 + (v - E)^2\right].
    \]
    Then the spectrum of $T$ is purely essential, of the form
    \[
        \sigma(T)  = \sigma_{\rm ess}(T) = [-b, -a] \cup [a, b],
    \]
    with
    \[
        \begin{cases}
            a & = \min \sigma \left( | T |(\bk = 0) \right) =  \frac{3}{\sqrt{2}} s \\
            b & = \max \sigma \left( | T |(\bk = 0) \right) = 3E.
        \end{cases}
    \]
\end{lemma}

At Fermi level $0$, this implies that $T$ is an insulator, unless $t = u = v$ (pristine phase).

\begin{proof}
    First we write $T(t,u,v) = t.S_1 + u.S_2 + v.S_3$, where $S_1$ only describes bonds with amplitude $t$ i.e. $S_1 = T(1,0,0)$, and so on. Note that $S_1^2 = S_2^2 = S_3^2 = I$, since no $t$--bond is connected to a different $t$--bond (and similarly with $u$ and $v$). This gives  
    \begin{equation} \label{eq:T2_with_Omega}
        T(t,u,v)^2 = \left( t^2 + u^2 + v^2 \right) I+tu \Omega_{12} + tv\Omega_{13} + uv\Omega_{23},
    \end{equation}
    where $\Omega_{12} = S_1S_2+S_2S_1$, $\Omega_{13} = S_1S_3+S_3S_1$ and $\Omega_{23} = S_2S_3+S_3S_2$. 
    
   \medskip
    
    The term $S_1S_2$ is the double-hop operator along $u$--$t$ bonds (in that order), hence $(S_1S_2)^3 = I$, as three double hops take any site around the hexagon with $t$ and $u$ bonds, back to its starting point. Note also that $S_1S_2$ is unitary. These remarks also hold for $S_1S_3$ and $S_2S_3$, so they all satisfy the simple operator equations $ MM^* = I$ and $M^3 = I$. Any such operator $M$ satisfies the following:
    \[ (M + M^*)^3 = 2 I+ 3(M+M^*),\]
    meaning that $M+M^*$ can only have eigenvalues $-1$ and $2$. Thus our previous terms $\Omega_{ij}$ are all lower bounded by $-1$. For non--negative $t,u,v$, this yields the value of $a$ set above:
   \[
        T(t,u,v)^2  \geq t^2+u^2+v^2 - (tu + tv + uv) = \frac{9}{2}s^2
    \]
and thus $|T(t,u,v)| \geq \frac{3}{\sqrt{2}} s$. The inequality is seen to be attained by computing the eigenvalues of $T(\bk = 0)$. In the pristine case ($s = 0$) we see the Dirac cones at $\bk = 0$ in $B_6$. We conclude that the Kekulé distortion opens a gap of size exactly $2\frac{3}{\sqrt{2}}s$. This is analogous to the one-dimensional Peierls case.

    \medskip
    
    Finally, note that $\Tr (T^2(\bk)) = 6 (t^2 + u^2 + v^2)$ is independent of $\bk$. In particular, the highest eigenvalue of $T(\bk)$ is reached at $\bk = 0$: at this point, the two other positive values takes their minimal value $\frac{3}{\sqrt{2}}s$. The result follows.
\end{proof}

%%%%%%%%%%%%%%
\subsubsection{The case of pristine graphene.} 
\label{ssec:pristine}
In the case where $t = u = v$, the operator $T$ is also $\LL_2$--periodic. This means that the operator $T(t,t,t)$ is also a Fourier multiplier on $\C^2$ with respect to the Brillouin zone $B_2$ associated to $\LL_2$, (see Figure~\ref{fig:brillouin}), with the classical formula
\[ T_2(t, \bk) := T(t, t,t, \bk) = t
\begin{pmatrix}
    0 & 1+\e^{-\ri \bk \cdot \ba_1}+\e^{-\ri \bk \cdot \ba_2} \\ 1+\e^{\ri \bk \cdot \ba_1}+\e^{\ri \bk \cdot \ba_2} & 0 
\end{pmatrix}.\]
Setting
\[
    m(\bk) := |1+\e^{\ri \bk \cdot \ba_1}+\e^{\ri \bk \cdot \ba_2}|,
\]
we can easily check that its spectrum is 
\[
    \sigma(T_2(t,\bk)) = \{\ \pm t.m(\bk) \ \}.
\]
The expression for $m$ yields the well-known band diagram for pristine graphene, with two distinctive Dirac cones shown in Figure~\ref{fig:bdiagram} at points $\bK, \bK'$ in $B_2$. At these points, $T_2(\bk = \bK)$ and $T_2(\bk = \bK')$ are the null matrices, hence have $0$ as an eigenvalue of multiplicity $2$.

\medskip

The Brillouin zone $B_6$ is exactly one third of $B_2$. This means that, studying $T(t,t,t)$ with respect to the $\LL_6$ lattice (instead of the $\LL_2$ lattice), we have the band folding formula
\[
    \sigma (T(\bk)) = \sigma(T_2(\bk)) \cup \sigma(T_2(\bk + \bb_1^*)) \cup \sigma(T_2(\bk + \bb_2^*)).
\]
In particular, after band folding, the two distinct Dirac cones on $B_2$ are both located at $\bk = \bnull$ in $B_6$. 
\medskip

We display the band diagram of $T(t,u,v)$ and $T(t,t,t)$ in Figure~\ref{fig:bdiagram}.

\begin{figure}[!ht]
    \centering
    \includegraphics[scale=0.7]{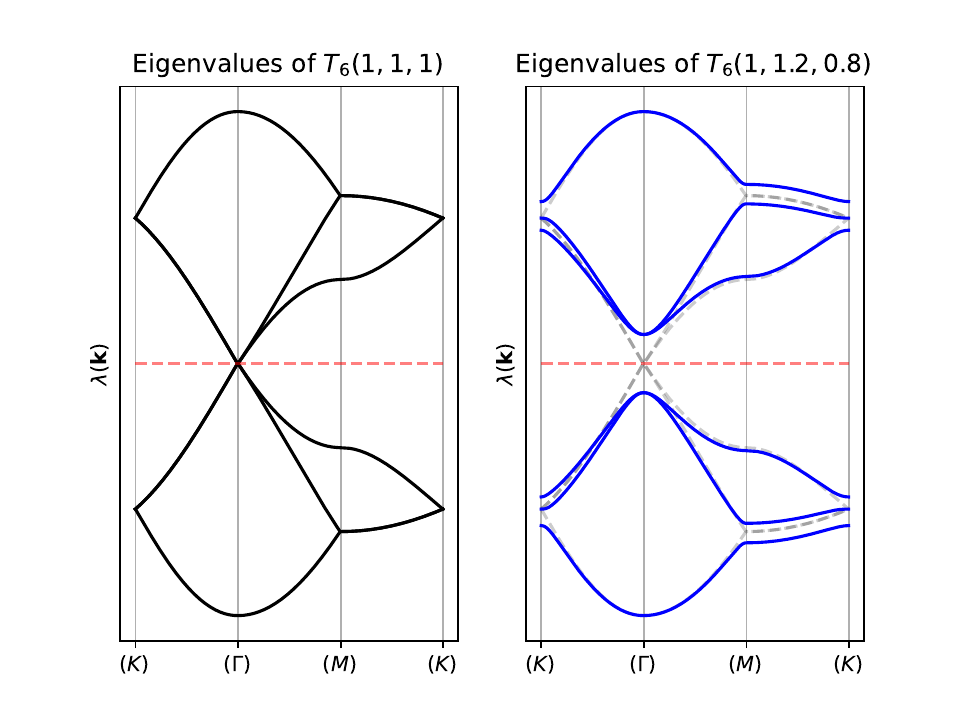}
    \caption{Band diagram cut of graphene over $B_6$ with six surfaces for each eigenvalue of $T(t,u,v,\bk)$, in the undistorted case (left, with two Dirac cones at central point $(\Gamma)$), and in the distorted case (right, with a bandgap).}
    \label{fig:bdiagram}
\end{figure}

%%%%%%%%%%%%%%%%%%%%%%%%%%%%%%
\section{Kekulé symmetry of minimizers}

We now prove Theorem~\ref{thm:Ksym_short}: a minimizer for our considered energy must retain a symmetry by which $t=u$ or $t=v$ or $u=v$, regardless of the positive value of $\mu$. In what follows, we again denote by
\[
E := \frac13 (t + u + v) \quad \text{and} \quad s^2 := \frac{1}{3} \left[ (t-E)^2 + (u - E)^2 + (v - E)^2\right]
\]
the average and variance of a triple $t,u,v$. We also introduce the variables
\[
    S := \sqrt{t^2 + u^2 + v^2}, \quad \text{and} \quad  W := tuv.
\]
Note that we have $S^2 = 3 (E^2 + s^2)$, hence, in particular, $S \ge \sqrt{3} E$. The following statement is a more precise version of Theorem~\ref{thm:Ksym_short}.

\begin{lemma} \label{thm:Ksym_long}
    For all $t,u,v \in \R$, there are $0 \le \tilde{t} \le \tilde{v} $ such that $\mce(t,u,v)~\geq~\mce(\tilde{t}, \tilde{t}, \tilde{v})$. These are chosen as follows: 
    \begin{enumerate}
        \item If $E \leq s/\sqrt{2}$, take $\tilde{v}:= S$ and $\tilde{t}:= 0$. 
        \item If $E >    s/\sqrt{2}$, take $\tilde{v} := E +s\sqrt{2}$ and $\tilde{t} := E -s/\sqrt{2}$. 
    \end{enumerate}
    The energy inequality is strict if $t,u,v$ is not Kekulé-symmetric. 
\end{lemma}

Note that having $ 0 \le \tilde{t} \le \tilde{v}$ is in line with the physical remark above, that a Kekulé configuration should have one third short bonds and two thirds long, and vice versa in hopping amplitudes. Also, if we take a flat positive triplet $(t,t,t)$, naturally $\tilde{t}$ and $\tilde{v}$ are both set to $t$, since a pristine configuration is also Kekulé-symmetric.

\begin{proof}[Proof of Lemma \ref{thm:Ksym_long}]

We treat the two cases separately. The condition $E > \frac{s}{\sqrt{2}}$ implies $E > 0$. This inequality also reads $3 E > S$, or   $ t + u + v > \sqrt{t^2 + u^2 + v^2}.$

\medskip

\underline{First case}. 

We use the following lemma. 
\begin{lemma} \label{lem:sphere_bound}
    For all $t,u,v$: 
    \[ \VTr  ~| T(t,u,v) | \leq \VTr  ~| T(S,0,0) |. \]
    There is equality if and only if $t,u,v$ has two zero values. 
\end{lemma}

In other terms: in the $3$-dimensional $(t,u,v)$ space, on the sphere centered at the origin with radius $S$, the quantum energy term has its minimum at points $(S,0,0)$, $(0,S,0)$, $(0,0,S)$.

\begin{proof}[Proof of Lemma \ref{lem:sphere_bound}]
We recall that if $f : \R^+ \to \R$ is convex/concave, then $A \mapsto \Tr(f(A))$ is also convex/concave on the set of positive semi-definite Hermitian matrices. From~\eqref{eq:T2_with_Omega}, we have
\begin{align*}
    \VTr ~|T |
    & = \VTr ~\sqrt{ T^2 }
     =  \VTr \sqrt{ t^2+u^2+v^2 + tu \Omega_{12} + tv \Omega_{13} + uv \Omega_{23} }
\end{align*}
where as before the $\Omega$ operators describe the three different double hops available in the $(t,u,v)$ configuration. Introduce the functional over $\mathbb{R}^3$
\[
    f: x,y,z \mapsto \VTr \left[ (t^2+u^2+v^2) I + x \Omega_{12} + y \Omega_{13} + z \Omega_{23} \right]^{1/2}.
\]
The function $f$ is strictly concave, and since the $\Omega$'s have null diagonal (hence null trace), we find $\nabla f(0) = 0$. Hence for all non-zero triplet $x,y,z$, we have $f(x,y,z) < f(0)$. We compute $f(0) = 6S$, which is the value of $\VTr \ |T(S,0,0)|$, concluding the proof of Lemma~\ref{lem:sphere_bound}.
\end{proof}

The previous Lemma means that replacing $(t,u,v)$ by $(S, 0,0)$ always lowers the kinetic quantum energy in $\mce$. As for the distortion energy, we have 
\[
    \begin{cases}
        (t-1)^2 + (u-1)^2 + (v-1)^2 & = S^2 + 3 - 6 E  \\
        (S-1)^2 + 1^2 + 1^2 & =  S^2 + 3 - 2S,
    \end{cases}
\]
hence the distortion energy is also lowered with this replacement whenever $6E \le 2S$, which is also $E \le \frac{s}{\sqrt{2}} $. This proves the first case of Theorem~\ref{thm:Ksym_long}.

\medskip

\underline{Second case}. From now on, we assume that $E > s/\sqrt{2}$. We will exhibit a lower bound on the quantum energy that is pointwise in $\textbf{k} \in B_6$. In what follows, we fix $\bk$ and shorten $T(t,u,v,\textbf{k})$ to $T$. 

\medskip

We are interested in the quantity
\[
    \Tr_{\C^6} \left( \sqrt{T^2} \right) = 2 \Tr_{\C^3} \left( \sqrt{A^* A} \right) = 2 \left( \sqrt{x_1} + \sqrt{x_2} + \sqrt{x_3}\right),
\]
where $x_1, x_2, x_3$ are the eigenvalues of the $3 \times 3$ Hermitian matrix $A^*A$. They are the roots of the characteristic polynomial 
\[
    \chi(X) := \chi_{A^* A}(X) := \det(XI_3 - A^* A).
\]
Let us denote the coefficients of $\chi(\cdot)$ by $\chi(X) = X^3 - \alpha X^2  + \beta X - \gamma$. Before we compute these coefficients, we make the remark that they must be {\em symmetric polynomials} in the variables $(t,u,v)$, since they play similar roles. In particular, they can be written solely with the variables $E$, $S$ (or $s$) and $W$. On the other hand, the distortion energy only depends on $E$ and $S$. Our main goal is therefore to prove that, as a function of $W$, with $E$ and $S$ fixed, the sum  $\sqrt{x_1} + \sqrt{x_2} + \sqrt{x_3}$ is {\em increasing} in $W$. This will prove the result thanks to the following elementary result.

\begin{lemma} \label{lem:maxtuv}
    Consider the circle $\mathscr{C} \subset \R^3$ of triplets $t,u,v$ sharing identical values of $E$ and $S$, it has center $(E,E,E)$ and radius $\sqrt{3}s$, namely
    \[
        \sC := \left\{ t,u,v \in \R^3 :  \frac13 (t + u + v)  = E, ~ \frac{1}{3} \left[ (t-E)^2 + (u - E)^2 + (v - E)^2 \right] =  s^2 \right\}.
    \]
    Then, the extrema of $W = tuv$ over $\mathscr C$ are Kekulé-symmetric: for $(t,u,v) \in \sC$,
    \[
    (E - s\sqrt{2}) \left( E + \frac{s}{\sqrt{2}} \right) \left( E + \frac{s}{\sqrt{2}} \right) \leq tuv \leq (E + s\sqrt{2})\left( E - \frac{s}{\sqrt{2}} \right) \left( E - \frac{s}{\sqrt{2}} \right).
    \]
\end{lemma}

\begin{proof}[Proof of Lemma~\ref{lem:maxtuv}]
The constraints imply that there exist $\lambda, \mu$ such that 
\[tuv = \lambda t^2 + \mu t = \lambda u^2 + \mu u = \lambda v^2 + \mu v, \]
 thus $t,u,v$ are roots of a single quadratic, meaning two of the three are equal. Computing Kekulé-symmetric points on $\sC$ yields the above expressions with $E\pm~s\sqrt2$, $E\pm s/\sqrt2$, which then form $\tilde{t}, \tilde{v}$ in Lemma~\ref{thm:Ksym_long}. 
\end{proof}

\medskip

Let us now compute the coefficients of $\chi$. First, we can easily check that 
\[
    \alpha = \Tr(A^* A) = 3 (t^2 + u^2 + v^2) = 3 S^2.
\]
Next, we have, 
\begin{align*}
    \gamma & = \det(A^* A) = | \det(A) |^2 = \left| t^3 + u^3 + v^3 - tuv \left( \re^{-\ri\bk \cdot \bb_1} + \re^{-\ri\bk \cdot \bb_2}  + \re^{\ri\bk \cdot (\bb_1 + \bb_2)} \right) \right|^2 \\
    & = \left| \frac{27}{2} E s^2 + W Z(\bk) \right|^2 ,
\end{align*}
where we set
\[
    Z(\bk) := 3 - \re^{-\ri \bk \cdot \bb_1} - \re^{- i \bk \cdot \bb_2}  - \re^{\ri \bk \cdot (\bb_1 + \bb_2)}.
\]
This is the only variable which depends on $\bk$. Finally, the last coefficient is more tedious to compute, although straightforward, and we find
\[
    \beta = 3S^4-\frac34(9E^2 - S^2)^2 + 6WE.\Re(Z(\textbf{k})).
\]

Recall that we are interested in the sum $\sqrt{x_1} + \sqrt{x_2} + \sqrt{x_3}$, where $x_1, x_2, x_3$ are the (positive) roots of $\chi(\cdot)$. It is unclear that this sum has an explicit expression in terms of the coefficients of $\chi$. However, this quantity solves a simple quartic equation. To see this, we note that
\[
    \left( \sqrt{x_1} + \sqrt{x_2} + \sqrt{x_3} \right)^2 = x_1 + x_2 + x_3 + 2 (\sqrt{x_1x_2} + \sqrt{x_2x_3} + \sqrt{x_3x_1})
\]
and that
\[
    (\sqrt{x_1x_2} + \sqrt{x_2x_3} + \sqrt{x_3x_1})^2 = x_1 x_2 + x_2 x_3 + x_3 x_1 + 2 \sqrt{x_1 x_2 x_3} \left( \sqrt{x_1} + \sqrt{x_2} + \sqrt{x_3} \right).
\]
Together with the fact that $x_1 + x_2 + x_3 = \alpha$, that $x_1 x_2 + x_2 x_3 + x_3 x_1 = \beta$ and that $x_1 x_2 x_3 = \gamma$, we obtain the implicit equation (we set $g := \sqrt{x_1} + \sqrt{x_2} + \sqrt{x_3}$):
\begin{equation} \label{eq:g_inter}
    \frac{1}{4}\left(g^2 - \alpha \right)^2 = \beta + 2\sqrt{\gamma} g.
\end{equation}

We consider $E$ and $S$ (and $s$) to have fixed values, and we see the coefficients $\beta$ and $\gamma$ depend on $W$ only ($\alpha$ is independent of $W$). They are of the form
\[
        \beta = \beta(W) =: \beta_0 + \beta_1 W, 
        \quad \text{and} \quad 
        \gamma = \gamma(W) =: \gamma_0 + \gamma_1 W + \gamma_2 W^2.
\]
Equation~\eqref{eq:g_inter} states that $g(W)$ is an intersection between the quartic $x~\mapsto~\frac{1}{4}(x^2~-~\alpha)^2$ (independent of $W$) and the line $x \mapsto \beta(W) +2\sqrt{\gamma(W)}x$. Specifically, it is the only intersection greater than $\sqrt{\alpha}$, due to the inequality
\[
g = \sqrt{x_1} + \sqrt{x_2} + \sqrt{x_3} \geq \sqrt{x_1 + x_2 + x_3} = \sqrt{\alpha}.
\]
Since $\beta_1 = 6 E \cdot \Re(Z) \ge 0$, the map $W \mapsto \beta(W)$ is nondecreasing. In addition, the map $W \mapsto \gamma(W)$ is also increasing for all $W \ge - \frac{\gamma_1}{2 \gamma_2}$.

\begin{lemma} \label{lem:bound_W_second_case}
    Under the condition $E > s/\sqrt{2}$, we have $W \ge - \frac{\gamma_1}{2 \gamma_2}$ over $\mathscr{C}$.
\end{lemma} 
This implies that the intersection $g(W)$ is increasing in $W$ in the second case. Thus the maximum of $g$ over $\mathscr C$ is attained for the same parameters as the maximum of $W$ over $\mathscr{C}$, and we conclude with Lemma~\ref{lem:maxtuv}. 

\medskip

It remains to prove Lemma~\ref{lem:bound_W_second_case}.
\begin{proof}[Proof of Lemma~\ref{lem:bound_W_second_case}]
First, we compute that
\[
    \frac{-\gamma_1}{2 \gamma_2} = - \frac{\Re(Z(\bk))}{| Z |^2(\bk)} \frac{27}{2} E s^2.
\]
Let us compute a lower bound of the first fraction, valid for all $\bk$. We have
\begin{align*}
    |Z|^2(\textbf{k}) & = 12 - 6\cos (\bk \cdot \bb_1) -6\cos (\bk \cdot \bb_2) - \cos (\bk \cdot (\bb_1+\bb_2)) \\  
    & \quad + 2\cos (\bk \cdot (\bb_1 - \bb_2)) + 2\cos (\bk \cdot (2\bb_1 + \bb_2)) + 2\cos (\bk \cdot (\bb_1 + 2\bb_2)) \\
    & = 6\left[ 3 - \cos (\bk \cdot \bb_1) - \cos (\bk \cdot \bb_2) - 6\cos (\bk \cdot (\bb_1+\bb_2)) \right] \\ 
    & \quad  - 2\left[ 3 - \cos (\bk \cdot (\bb_1 - \bb_2)) -\cos (\bk \cdot (2\bb_1 + \bb_2)) -\cos (\bk \cdot (\bb_1 + 2\bb_2)) \right] \\
    & = 6 \Re (Z(\bk)) - 2 \Re(\widetilde{Z}(\bk)),
\end{align*}
where we set
\[
    \widetilde{Z}(\bk) := 3 - \re^{\ri \bk\cdot (\bb_1 - \bb_2)} -  \re^{- \ri \bk\cdot (2 \bb_1 + \bb_2)}  - \re^{- \ri \bk\cdot (\bb_1 + 2 \bb_2)}.
\]
We have $\Re(\widetilde{Z}(\bk)) \ge 0$, giving the bound
\[
    \frac{\Re(Z(\bk))}{| Z |^2(\bk)} = \frac{1}{6 - 2 \frac{\Re(\widetilde{Z}(\bk))}{\Re(Z(\bk))} } \ge \frac{1}{6},
    \quad \text{hence} \quad
    \frac{\gamma_1}{2 \gamma_2} \ge \frac{9}{4}E s^2.
\]
This bound is sharp because $\Re(\widetilde{Z}(\bk))$ vanishes at points where $\Re(Z(\bk))$ does not, specifically at $\pm (\bb_1^* + \bb_2^*)/3$. On the other hand, according to Lemma~\ref{lem:maxtuv}, we have
\[
    W \ge  (E - s\sqrt{2}) \left( E + \frac{s}{\sqrt{2}} \right) \left( E + \frac{s}{\sqrt{2}} \right) = E^3 -\frac32 Es^2 - \frac{1}{\sqrt{2}}s^3.
\]
Together with the condition $E > s/\sqrt{2}$, we obtain 
\begin{align*}
    W +  \frac{\gamma_1}{2 \gamma_2}
    & \geq \left(E^3 -\frac32 Es^2 - \frac{1}{\sqrt{2}}s^3\right) + \frac94 Es^2
     = E^3 - \frac{1}{\sqrt{2}}s^3 + \frac34Es^2 \\
    & \geq s^3 \left( \frac{1}{2\sqrt{2}} - \frac{1}{\sqrt{2}} + \frac{3}{4\sqrt{2}} \right) = \frac{1}{4\sqrt{2}}s^3 \geq 0.
\end{align*}
\end{proof}

This ends the proof of Theorem~\ref{thm:Ksym_long}, hence of Theorem~\ref{thm:Ksym_short}.
\end{proof}

%%%%%%%%%%%%%%%%%%%%%%%%%%%
\section{Translation-invariance for large $\mu$} \label{sec:trans_invar}

In this section, we prove the second part of Theorem~\ref{thm:phasetrans}, namely that for $\mu$ large enough, the optimizer is the pristine phase. Let us first compute the optimizer among all pristine phases, that is among all triplets of the form $(t,t,t)$. According to Lemma~\ref{thm:Ksym_long}, this optimizer $t_*$ is positive, so we assume $t \ge 0$ from now on.

\medskip

\subsection{Critical point among pristine configurations}

Henceforth, define the pristine Hamiltonian $H := T_2(t=1)$ from Section~\ref{ssec:pristine}. For any pristine phase, the energy per Carbon atom is
\[
\cE(t) =  - t \cdot \VTr \ | H | + \frac{\mu}{2} \cdot \frac{3}{2} (t-1)^2.
\]
We used here translation invariance by the $\LL_2$--lattice, and the fact that there are two carbon atoms in $\Gamma_2$, and three $t$--bonds. The expression is quadratic in $t$, hence the unique minimizer among pristine configurations is given by: 
\begin{equation} \label{eq:formula_t*}
\boxed{t_* = 1 + \frac{2}{3\mu}\VTr | H | = 1 + \frac{2}{3\mu} \fint_{B_2} m(\bk) \dee \bk,}
\end{equation}
where we recall that $m(\bk) := | 1 + \re^{\ri \bk \cdot \ba_1 } + \re^{\ri \bk \cdot \ba_2} |$, see Section~\ref{ssec:pristine}.

\subsection{Translation invariance for $\mu$ large enough} Take $L > 1$ an integer. We now prove that for $\mu$ large enough,  the optimal pristine configuration is in fact optimal among all configurations with $L$-periodicity with respect to the $\LL_2$ lattice (recall that the Kekulé texture is a subcase of $L=3$). In addition, our argument works {\em mutatis mutandis} for any quantum energy of the form $\VTr~f(T^2)$, with $f$ convex and $C^1$ over $\R^+$ (here $f(t) = - \sqrt{t}$).

Recall that for such $L$-periodic configurations $\textbf{t}$, the energy per Carbon atoms is
\[
    \cE(\textbf{t}) :=  - \VTr \ |T(\textbf{t})| + \frac{\mu}{2} \left( \frac{1}{2L^2}  \right) \sum_i (t_{i,1} - 1)^2 + (t_{i,2} - 1)^2 + (t_{i,3} - 1)^2,
\]
where $i$ numbers each of the 2-atom cells in a given $L\cdot \Gamma_2$ supercell, each 2-atom cell has three bonds with amplitudes $t_{i,1}, t_{i,2}, t_{i,3}$. 

\medskip

In what follows, we set
\[
    \bt =: \bt_* + \bh, \quad \varepsilon := \langle \bh \rangle = \frac{1}{3 L^2} \sum_{i,n} h_{i,n}
\]
where $\bt_*$ is the optimal pristine configuration. Here, $\varepsilon$ is the mean of $\bh$, so the mean of $\bt$ is $\langle \bt \rangle = t_* + \varepsilon$. Since $\bt \mapsto T(\bt)$ is linear in $\bt$, we have

\begin{align*}
\cE(\textbf{t}) 
& = - \VTr \left(  \sqrt{ \left[ T(\bt_*) + T(\bh) \right]^2} \right) + \frac{\mu}{2} \left( \frac{1}{2L^2}  \right) \sum_{i,l} (t_* + h_{i,l} - 1)^2 \\ 
& = - \VTr \left(  \sqrt{ \left[ T(\bt_*) + T(\bh) \right]^2} \right) + \frac{\mu}{2} \left[ \frac32(t_* - 1)^2 + 3 (t_* - 1) \varepsilon + \frac{1}{2L^2} \sum_{i,n} h_{i,n}^2 \right] \\ 
& = - \VTr \left( \sqrt{  t_*^2H^2 + t_* \left[ H T(\bh) + T(\bh) H \right] + T(\bh)^2 } \right) \\
& \quad \quad \quad + \frac{\mu}{2} \frac32(t_* - 1)^2 +  \varepsilon \VTr |H| + \frac{\mu}{2}\frac{1}{2L^2} \sum_{i,n} h_{i,n}^2~~, 
\end{align*}
where in the last line the $\varepsilon \VTr |H|$ term appears by recalling the expression of $t_*$ at \eqref{eq:formula_t*}, and furthermore the $(t_*-1)^2$ term is the pristine elastic energy. 

\medskip

The main idea of the proof is to use again the (strict) convexity of $f(t) := - \sqrt{t}$, whose derivative is $f'(t) = - \frac{1}{2 \sqrt{t}}$, so that $A \mapsto \VTr ( f(A))$ is convex as well. For any self-adjoint $B$ and positive $A$, we have
\[
    -\VTr\sqrt{A+B}\geq-\VTr\sqrt{A}- \VTr \left( \frac{1}{2\sqrt{A}}B \right).
\]

We apply the previous inequality with 
\[
    A = t_*^2 H^2, \quad \text{and} \quad
    B = t_* \left( H T(\bh) + T(\bh) H \right) + T(\bh)^2. 
\]
Note that $H$ is not invertible. However, we know that the band structure of $|H|$ is described by the function $m(\bk) = |1+\e^{\ri\textbf{k}\cdot \textbf{a}_1}+\e^{\ri\textbf{k}\cdot \textbf{a}_2}|$ over $B_2$, which is zero only at the Dirac points, where it is conic: $m(\textbf{K}+\bk) \sim_{\bk\to 0} \sqrt{3}/2|\bk|$. Hence its inverse is integrable over the 2-dimensional Brillouin zone, meaning that $\VTr \frac{1}{|H|}B$ is finite for any bounded operator $B$, and the above inequality still holds. Computing the inequality yields
\begin{equation}\label{ineq:Trconv}
- \VTr | T(\bt) | \ge - t_* \VTr | H | - \VTr \left( {\rm sgn}(H) T(\bh) \right) -  \frac{1}{2 t_*} \VTr \left( \frac{1}{| H |} T(\bh)^2 \right).
\end{equation}

The strict convexity implies that this inequality is only attained when $T(\textbf{t})^2$ coincides with $ (t_*H)^2$. This means that all double hops on the honeycomb lattice have amplitude $t_*^2$, and since each site has three neighbours, we can conclude all amplitudes are actually equal, with either $\textbf{t} = \textbf{t}_*$ or $\textbf{t} = -\textbf{t}_*$. 

\medskip

Let us now compute the middle term in \eqref{ineq:Trconv}. By invariance with respect to the $2$-atom cell translations and $\pi/3$ rotations, and by linearity of $\bh \mapsto T(\bh)$, we have
\[
    \VTr \left( {\rm sgn}(H) T(\bh) \right) = \VTr \left( {\rm sgn}(H) T(\langle \bh \rangle) \right) = \varepsilon \VTr  \left( {\rm sgn}(H) H \right) = \varepsilon \, \VTr \, | H |.
\]
Applying our previous computations of $\cE$, we deduce that 
\begin{equation} \label{ineq:en_pre_kag}
\cE(\bt) \ge \cE(\bt_*) + \left( \frac{\mu}{2}\frac{1}{2 L^2} \sum_{i,n} h_{i,n}^2   -  \frac{1}{2 t_*} \VTr \left( \frac{1}{ | H |} T(\bh)^2 \right) \right).
\end{equation}

Note that having cancelled the linear term in $\varepsilon$ reflects the fact that $t_*$ is the optimizer among pristine configurations. To continue the computation, we write the double hop operator $T(\bh)^2$ as
\[
    T(\bh)^2 = \sum_{i,n} h_{i,n}^2 I_i + \sum_{\ps{(i,n)}{(j,m)}} h_{i,n} h_{j,m} J_{(i,n),(j,m)}
\]
where $I_i$ describes double hops from cell $i$ back to $i$, and $J_{(i,n),(j,m)}$ the double hop along adjacent bonds $(i,n)$ and $(j,m)$. There are $6L^2$ such terms (6 double hops per 2-cell). 

\medskip

Again, invariance by translations and rotations of the model implies the two following identities. First, since $\sum_i I_i = I$, we have
\[
    \VTr \frac{1}{|H|} I_i = \frac{1}{L^2} \VTr \frac{1}{|H|}.
\]
Second, computing $H^2 = 3 I + \sum_{\ps{(i,n)}{(j,m)}} J_{(i,n),(j,m)}$ yields
\[
    \VTr \frac{1}{|H|} H^2 = 3 \VTr \frac{1}{|H|} + 6L^2 \ \VTr \frac{1}{|H|} J_{(i,n),(j,m)}
\]
hence
\[
\VTr \frac{1}{|H|} J_{(i,n),(j,m)} = \frac{1}{6L^2} \left(\VTr |H| -  3 \VTr \frac{1}{|H|} \right),
\]
independently of $\langle (i,n), (j,m)\rangle$. This gives
\medskip
\begin{align*}
    \VTr \left( \frac{1}{|H|} T(\bh)^2  \right)
    = & ~ \VTr \frac{1}{|H|} \cdot \frac{1}{L^2} \sum_{i,n} (h_{i,n})^2 \\
    & \quad - \left( 3 \VTr \frac{1}{|H|} - \VTr |H|  \right) \cdot \frac{1}{6L^2} \sum_{\ps{(i,n)}{(j,m)}} h_{i,n} h_{j,m}  .
\end{align*}

Note that $3 \VTr \frac{1}{|H|} - \VTr |H|$ is positive, by the following Cauchy-Schwarz inequalities, using $\VTr (H^2) = 3 $:
\[
1 = \VTr (I) < \VTr |H| \cdot \VTr \frac{1}{|H|} < \VTr (H^2) \frac{1}{\VTr |H|} \cdot \VTr \frac{1}{|H|} = 3 \frac{1}{\VTr |H|} \cdot \VTr \frac{1}{|H|}. 
\]

To conclude, we claim the following.

\begin{lemma}
    We have
    \[
        \sum_{\ps{(i,n)}{(j,m)}} h_{i,n}h_{j,m} \geq -\sum_{i,n} h_{i,n}^2 
    \]
\end{lemma}
 
This is essentially a generalisation of the bound used in the proof of Lemma~\ref{lem:spectrum_T}. 
 
\begin{proof}
    We introduce the \emph{line graph} for our periodic honeycomb sample: this line graph is a Kagome (or trihexagonal) lattice. Denoting by $A$ the corresponding Kagome adjacency matrix, we have
    \[
        \sum_{\ps{(i,n)}{(j,m)}} h_{i,n}h_{j,m} = \frac12 \ps{\bh}{A\bh},
    \]
    where the 1/2 factor accounts for each double hop being counted twice. It is well-known for the Kagome lattice that $A$'s lowest eigenvalue is a flat band with value $-2$, see for instance~\cite{Mie-91} (we include the computation in Appendix~\ref{appendix:kagome} for completeness), and the result follows.
\end{proof}

\medskip

Setting $\delta^2 :=  \langle \bh^2 \rangle :=  \frac{1}{3 L^2} \sum_{i,n} h_{i,n}^2$, we can continue the lower bound from \eqref{ineq:en_pre_kag}:
\begin{align*}
     \cE(\bt) 
     & \ge \cE(\bt_*)   + \frac{\mu}{2}\frac{3}{2} \delta^2   -  \frac{1}{2 t_*} \left[ 3\delta^2 \VTr  \frac{1}{ | H |} + \frac12 \delta^2 \left(3 \VTr \frac{1}{|H|} - \VTr |H| \right)   \right]  \\
     & \ge \cE(\bt_*) + \frac{3\delta^2}{2t_*} \left[ \frac{\mu}{2} t_* - \frac32\VTr \frac{1}{|H|} + \frac16 \VTr |H| \right] \\
     & \ge \cE(\bt_*) + \frac{3\delta^2}{4t_*} \left[ \mu+ \VTr |H|  - 3\VTr \frac{1}{|H|} \right].
\end{align*}
In the last line, we used the expression of $t_*$ in~\eqref{eq:formula_t*}. We conclude that for $\mu$ large enough, the term in brackets is positive, so the undistorted configuration $\textbf{t}_*$ is the unique minimizer of $\mce$ in the space of $L$-periodic configurations, for any $L > 1$. Uniqueness is guaranteed simply by the equality case in the convexity argument at \eqref{ineq:Trconv}, which is $\textbf{t} = \textbf{t}_*$ or $\textbf{t} = -\textbf{t}_*$, and we know that $\cE( \textbf{t}_* ) < \cE( -\textbf{t}_* )  $. 

\medskip

If we denote by $\mu_c'$ the largest critical $\mu$ above which the pristine case is optimal, we find an upper bound on $\mu_c'$ of $3\VTr \frac{1}{|H|} -  \VTr |H| $, which numerically yields
\[
\boxed{\mu_c' \leq \fint_{B_2} \left( \frac{3}{m(\textbf{k})} - m(\textbf{k}) \right)\dee\textbf{k} \simeq 1.114. } 
\]

%%%%%%%%%%%%%%%%%%%%%%%%%%%%%

\section{Loss of translation-invariance for small $\mu$}

We now show by a perturbation argument that for a sufficiently small parameter $\mu$, minimizers of $\cE$ are \textit{not} of the form $(t,t,t)$, i.e. there must be a breaking of translation-invariance over the honeycomb lattice by way of a Kekulé distortion. For this proof, we compute the Hessian of $\cE$ around the critical point $(t_*, t_*, t_*)$, and prove that for $\mu$ small enough, this Hessian is not positive.

\medskip

First, we recall that $T(t_*, t_*, t_*) = t_* H$ with $H = T(1,1,1)$. The spectrum of $H$ has been studied in Section~\ref{ssec:pristine} using the $2$-atom unit cell. However, in this section, we rather use its spectral properties on the $6$-atom unit cell. By band folding, the $6$ eigenvalues of $H(\bk)$ for $ \bk \in B_6$ are given by
\begin{align*}
    \lambda_1(\bk) = m(\bk), \quad
    \lambda_2(\bk) = m(\bk+\mathbf{b_1^*}), \quad
    \lambda_3(\bk) = m(\bk+\mathbf{b_2^*}), \\
    \lambda_4(\bk) = -\lambda_1(\bk), \quad 
    \lambda_5(\bk) = -\lambda_2(\bk), \quad 
    \lambda_6(\bk) = -\lambda_3(\bk).
\end{align*}

Recall that $m(\bk) = |1 + \re^{\ri \bk\cdot\ba_1}+ \re^{\ri \bk\cdot\ba_2}|$, and set the phase 
\[
\e^{i\theta(\bk)} = (1 + \re^{\ri \bk\cdot\ba_1}+ \re^{\ri \bk\cdot\ba_2})/m(\bk).
\]
The corresponding eigenvectors and eigenprojectors, noted $x_i(\bk)$ and $P_i(\bk)$, are naturally related to the Fourier modes. By band folding, and keeping track of phases in $B_6$ (see also our convention in Figure~\ref{fig:honeycomb1}), we find their expression given by 
\begin{equation*}
     x_1(\bk) = \frac1{\sqrt{6}}\begin{pmatrix}
        1, & 
        \e^{\ri\bk\cdot \ba_1},& 
        \e^{\ri\bk\cdot( \ba_1-\ba_2)},& 
        \e^{\ri\theta(\bk)},& 
        \e^{\ri\bk\cdot(\ba_1-\ba_2)+i\theta(\bk)},& 
        \e^{-\ri\bk\cdot\ba_2+i\theta(\bk)}
    \end{pmatrix}^T, 
\end{equation*}
then $x_2(\bk) = x_1(\bk+\mathbf{b_1^*})$ and $x_3(\bk) = x_1(\bk+\mathbf{b_2^*})$. The expressions of $x_4, x_5, x_6$ are identical by replacing phases $\e^{\ri\theta}$ by $-\e^{\ri\theta}$, as a result of the similarity $T \simeq -T$. We now prove the first point in Theorem~\ref{thm:phasetrans} by the second-order energy expansion detailed below. 

\begin{proposition}
    Set
    \[
        c(\bk) := 
            \frac{m(\bk + \mathbf{b_2^*})^2 \Theta_1(\bk)}{m(\bk) + m(\bk+\mathbf{b_1^*})} + \frac{m(\bk + \mathbf{b_1^*})^2 \Theta_2(\bk)}{m(\bk) + m(\bk+\mathbf{b_2^*})}
        + \frac{m(\bk)^2 \Theta_3(\bk)}{m(\bk + \mathbf{b_1^*}) + m(\bk+\mathbf{b_2^*})}.
    \]
    where $\Theta_1(\bk) = |e^{\ri\theta(\bk)}e^{-\ri\theta(\bk + \mathbf{b_2^*})} - e^{\ri\theta(\bk + \mathbf{b_2^*})}e^{-\ri\theta(\bk + \mathbf{b_1^*})}|^2$, and $\Theta_2$, $\Theta_3$ are obtained by cycling of indices. Then we have, for all $t > 0$,
    \begin{equation}\label{eq:energy_taylor}
    \cE(t + 2h, t - h, t - h) = \cE(t, t, t) + h^2 \left( 9 \mu - \frac{1}{t} \fint_{B_6} c(\bk) \rd \bk \right) + o( h^2).
    \end{equation}
 
    Considering in particular the case $t = t_*$, we obtain that for all $\mu < \mu_c$ there is $h > 0$ so that $\cE(t_*+ 2h, t_* -h, t_*-h) < \cE(t_*, t_*, t_*)$ , with
    \[
        \boxed{\mu_c := \frac19 \fint_{B_6} c(\bk) \rd \bk - \frac23 \fint_{B_2} m(\bk) \rd \bk.}
    \]
\end{proposition}

The rest of the section is the proof of this Proposition.

\subsection{Eigenvalue expansion of $T(1,1,1, \bk)$}

We consider the quantum kinetic energy. First, we note that
\begin{align*}
    T(t + 2h, t - h, t - h) & = (t - h) T(1 + \eta, 1, 1), \quad \text{with} \quad \eta =  \frac{3h}{t - h} \\
    & = (t - h) \left[ H + \eta S_1 \right],
\end{align*}
where $S_1$ is the adjacency matrix corresponding to the $t$--bonds only. As seen previously, this operator commutes with $\LL_6$--translations, hence can be studied with the $\LL_6$ Fourier transform. Note that
\[
    S_1(\bk) = \begin{pmatrix}
        0 & I_3 \\
        I_3 & 0
    \end{pmatrix},
\]
is independent of $\bk \in B_6$. In the following expansion, we take $\bk \neq 0$, by knowledge of function $m$, this ensures that $H(\bk)$ has no zero eigenvalues. Similarly, assuming $|\eta| < 1$,  the band gap seen in Lemma~\ref{lem:spectrum_T} shows that $|T(1+\eta,1,1)| \geq |\eta|$ (since the standard deviation of $(1+\eta, 1,1)$ is $|\eta| \sqrt{2}/3$), thus $H(\bk) + \eta S_1$ has no zero eigenvalues. 
\medskip

Recall that since $H(\bk)$ and $S_1$ have even spectra,  
\[
    \frac12 \Tr | H(\bk) + \eta S_1 | = \Tr \left[ H(\bk) + \eta S_1  \right]_+ = \frac{1}{2 i \pi} \int_\sC \frac{z}{z - H(\bk) - \eta S_1} \rd z,
\]  
where we take the contour $\sC$ to be a rectangle in $\mathbb{C}$ with corners $iM, (1+i)M, (1-i)M, -M$ for very large fixed $M$, so that $\sC$ encloses the positive eigenvalues that interest us, with no eigenvalues on $\sC$ itself since $\bk \neq 0$. The resolvent formula then shows that

\begin{align*}
   \frac12 \Tr | H(\bk) + \eta S_1 |
    = & \frac12 \Tr | H (\bk) | +  \eta \Tr (\mathbf{1}_{H(\bk) \geq 0} S_1) \\
    & + \eta^2\sum_{1 \leq i \leq 3, 4 \leq j \leq 6} \frac{1}{\lambda_i(\bk) - \lambda_j(\bk)} \Tr \left[ P_i (\bk) S_1 P_j(\bk) S_1 \right]   + \eta^2 R(\eta, \bk) 
\end{align*} 
with the remainder term
\[
R(\eta, \bk) := \eta \Tr \left( \frac{1}{2\pi i} \int_\sC z\left(\frac{1}{z-H} S_1 \right)^3 \frac{1}{z-H - \eta S_1} \rd z \right).
\]
For fixed $\mathbf{k} \neq 0$, we see that $R(\eta, \bk) \to 0$ as $\eta \to 0$, again because the eigenvalues of $H+\eta S_1$ stay within $\sC$ for all small $\eta$, thus the contour integral in the expression of $R$ is finite at $\eta = 0$.\medskip

\noindent The first-order term in the expansion computes as
\[
    \Tr (\mathbf{1}_{H(\bk) \geq 0} S_1) = \frac13 \Tr H(\bk)_+. 
\]
For the second-order term, a computation, detailed in Appendix~\ref{appendix:order2term} shows that
\[
    \sum_{1 \leq i \leq 3, 4 \leq j \leq 6} \frac{1}{\lambda_i(\bk) - \lambda_j(\bk)} \Tr (P_i(\bk) S_1 P_j(\bk) S_1) = \frac{1}{18} c(\bk).
\]

Notice in the expression of $c$ that the only term that diverges at $\mathbf{k} = 0$ is that with denominator $m(\mathbf{k} + \mathbf{b_1^*}) + m(\mathbf{k}+\mathbf{b_2^*})$. Since asymptotically the Dirac cones $1/(m(\mathbf{k} + \mathbf{b_1^*})+ m(\mathbf{k}+\mathbf{b_2^*}))$ are dominated by $1/|\mathbf{k}|$, we see that $c$ is integrable over the two-dimensional Brillouin zone $B_6$. 

\subsection{Second order expansion of the energy}

Let us prove that $R(\eta, \bk)$ is integrable in $\bk$, and that its integral goes to $0$ as $\eta \to 0$. Let $\gamma : [0,1] \to \C$ be the contour map, then using the fact that $S_1$ is bounded, we find
\[ |R(\eta, \bk)| \lesssim \int_0^1 \frac{|\eta| |\gamma(t)|}{\dist \left[\gamma(t), \sigma(H(\bk))\right]^{3}} ~~\frac{1}{ \dist \left[\gamma(t), \sigma(H(\bk) + \eta S_1)\right] } \dee t .\] 

Take $\sC_1$ to be the segment of the contour between $-iM$ and $iM$, and $\sC_2$ the rest of the contour. For the part with $\sC_2$, we can simply bound, for all $|\eta| \le 1$,
\[
    \forall z \in \sC_2, \quad  \dist \left[z, \sigma(H(\bk))\right] \geq M
    \quad \text{and} \quad
    \dist \left[ z, \sigma(H(\bk) + \eta S_1) \right] \geq M,
\]
So the part on $\sC_2$ is uniformly bounded in $\bk$ and $\eta$, hence integrable in $\bk \in B_6$. This integral goes to $0$ as $\eta \to 0$ thanks to the $\eta$ factor in the numerator. We now bound the part on $\sC_1$. We write $z \in \sC_1$ as $z = iw$ with $-M \leq w \leq M$. Then
\[
    \dist \left[z, \sigma(H(\bk))\right] = \sqrt{w^2 + \lambda_0(\bk)^2}
    \quad \text{and} \quad
    \dist \left[z, \sigma(H(\bk) + \eta S_1)\right] =  \sqrt{w^2 + \lambda_0^\eta(\bk)^2},
\] 
where $\lambda_0(\bk)$ and $\lambda_0^\eta(\bk)$ are the lowest positive eigenvalues of $H(\bk)$ and $H(\bk) + \eta S_1$ respectively ($\lambda_0(\textbf{k})$ is either $m(\bk+\mathbf{b_1^*})$ or $m(\bk+\mathbf{b_2^*})$). 

\medskip

We have already obtained lower bounds on these eigenvalues: by virtue of the Dirac cones, we have $\lambda_0(\bk) \geq C|\bk|$, and we proved in Lemma~\ref{lem:spectrum_T} and $\lambda_0^\eta(\bk) \geq |\eta|$ by the aforementioned bandgap. Thus we have the lower bounds: 
\[
    \dist (z, \sigma(H(\bk))) \geq \sqrt{w^2 + C|\bk|^2} 
    \quad \text{and} \quad
    \dist (z, \sigma(H(\bk) + \eta S_1)) \geq |\eta|.
\]
The part on $\sC_1$ is then bounded by
\begin{equation*}
    \int_{-M}^M \frac{|w|}{ (w^2 + C|\bk|^2)^{3/2} } \frac{|\eta|}{|\eta|}\dee w
    \le \int_{-\infty}^\infty \frac{|w|}{ (w^2 + C|\bk|^2)^{3/2} }  \dee w
    =  \frac{2}{\sqrt{C} | \bk |},
\end{equation*}
which is integrable in $\bk$ on the 2-dimensional Brillouin-zone $B_6$. So $| R(\eta, \bk) |$ is bounded by an integrable function, independent of $\eta$. In addition, we have the pointwise convergence $R(\eta, \bk)  \to 0$ as $\eta \to 0$. The dominated convergence theorem then implies that $\int_{B_6} R(\eta, \bk) \rd \bk \to 0$ as $\eta \to 0$.

\medskip

Altogether, we proved that the quantum term of the perturbed energy has the expansion 
\[
     \VTr | H + \eta S_1 |
    = \VTr | H  | +  \eta \frac13 \VTr | H | + 2 \eta^2 \frac{1}{18} \fint_{B_6} c(\bk) \rd \bk + o(\eta^2).
\]
When rewritten in $t$ and $h$ using the linearity of $T$, the first-order term disappears (this is because of symmetry in the $t,u,v$ variables), and we get
\[
 \VTr~ |T(t+2h, t-h, t-h)| = \VTr~ |T(t,t,t) | + \frac{h^2}{t} \fint_{B_6} c(\bk) \rd \bk + o(h^2)
\]
On the other hand, for the distortion energy, we have
\[
    \frac{\mu}{2} \left[ (t + 2 h - 1)^2 + 2 (t - h - 1)^2 \right] = \frac{\mu}{2} 3 (t - 1)^2 + 9 \mu h^2
\]
Altogether, we end up with
\[
    \cE(t+2h,t-h,t-h) = \cE(t,t,t) + h^2\left(9\mu - \frac1t \fint_{B_6} c(\bk) \dee \bk \right) + o(h^2),
\]
and the result follows. Thus the Kekulé perturbation lowers the energy when the coefficient $9\mu - \frac1t \fint_{B_6} c(\mathbf{k}) \dee \mathbf{k}$ is negative. We apply this in the case of the optimal pristine configuration $ t_* = 1 + \frac{2}{3\mu}\VTr |H|$, yielding an expression of $\mu_c$:

\begin{equation*}
\boxed{\mu_c = \frac{1}{9}\fint_{B_6} c(\mathbf{k})\dee \mathbf{k} - \frac23 \fint_{B_2}  m(\mathbf{k}) \dee \mathbf{k}.}
\end{equation*}

Numerically, this constant has a value of approximately $0.888$. Since this value follows from our Hessian computation, we expect it to be the true phase transition threshold for the Peierls energy, i.e. the local behaviour around the optimal pristine configuration should be indicative of the global picture.

\underline{Acknowledgements}: we collectively thank Rupert Frank and Michael Loss for stimulating discussions related to this work.

%%%%%%%%%%%%%%%%%%%%%%

\appendix

\section{Appendix}

\begin{figure}
    \centering
    \includegraphics[scale=0.3]{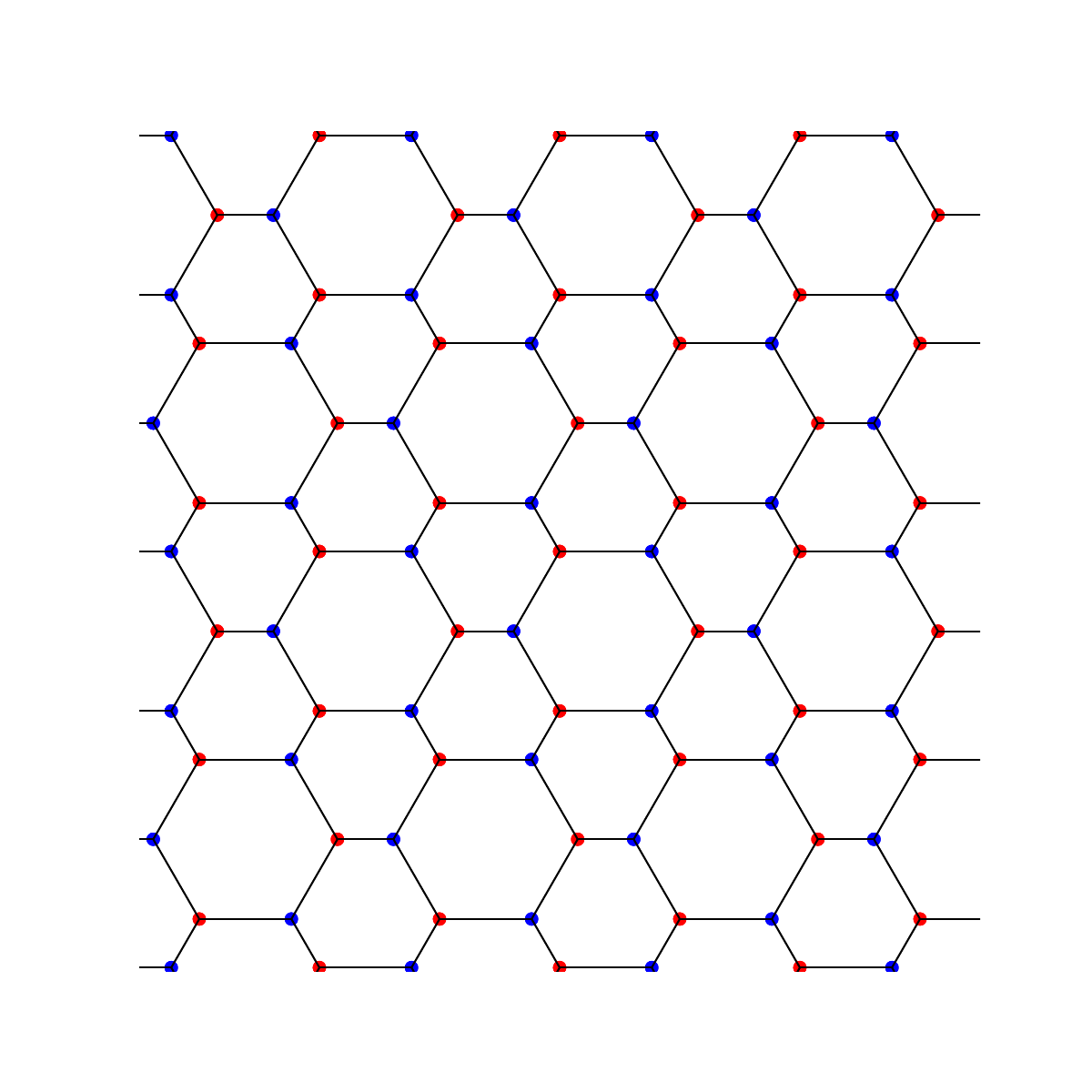}
    \caption{Exaggerated representation of the Kekulé O-type distortion, with a characteristic spread of both regular and benzene-like hexagons.}
    \label{fig:kekfig}
\end{figure}

\subsection{Computation of second-order term in~\eqref{eq:energy_taylor}} \label{appendix:order2term} ~\\

We seek to compute $\sum_{1 \leq i \leq 3, 4 \leq j \leq 6} \frac{1}{\lambda_i(\bk) - \lambda_j(\bk)} \Tr (P_i(\bk) S_1 P_j(\bk) S_1)$, that is 

\begin{equation} \label{eq:order2sum}
\sum_{1 \leq i \leq 3, 4 \leq j \leq 6} \frac{1}{\lambda_i(\bk) - \lambda_j(\bk)} |\ps{x_i(\bk)}{S_1 x_j(\bk)}|^2.
\end{equation}

To do so, we first cancel out a number of terms by symmetry. Recall that $S_1$ is shorthand for $T(1,0,0, \bk)$, and let $\tau$ be translation by the vector  $\mathbf{a}_1-\mathbf{a}_2$, with $\tau(\bk)$ its $B_6$ Brillouin space counterpart, this translation handily conjugates between different eigenvectors, as summarised here: 

\begin{align*}
    \tau(\bk)T(1,0,0, \bk)\tau(\bk)^* 
    & = T(0,1,0, \bk), \\
    \tau(\bk)^*T(1,0,0, \bk)\tau(\bk) 
    & = T(0,0,1, \bk), \\
    \tau(\bk)x_1(\bk) = \e^{\ri\bk\cdot(-\ba_1+\ba_2)}x_1(\bk), \ \ 
    & \tau(\bk)x_4(\bk) = \e^{\ri\bk\cdot(-\ba_1+\ba_2)}x_2(\bk),\\
    \tau(\bk)x_2(\bk) = j^2\e^{\ri\bk\cdot(-\ba_1+\ba_2)}x_2(\bk), \ \ 
    & \tau(\bk)x_5(\bk) = j^2\e^{\ri\bk\cdot(-\ba_1+\ba_2)}x_5(\bk),\\
    \tau(\bk)x_3(\bk) = j\e^{\ri\bk\cdot(-\ba_1+\ba_2)}x_3(\bk), \ \ 
    & \tau(\bk)x_6(\bk) = j\e^{\ri\bk\cdot(-\ba_1+\ba_2)}x_6(\bk).
\end{align*}

Then we can notice that
\[ \ps{x_1(\bk)}{T(1,0,0,\bk)x_4(\bk)} = \ps{x_1(\bk)}{T(0,1,0,\bk)x_4(\bk)} = \ps{x_1(\bk)}{T(0,0,1,\bk)x_4(\bk)} \]
\[ = \frac13\ps{x_1(\bk)}{H(\bk)x_4(\bk)} \]
where we used that $H = T(1,0,0,\bk) + T(0,1,0,\bk) + T(0,0,1,\bk)$. Since $\ps{x_1(\bk)}{H(\bk)x_4(\bk)}$ vanishes, this term is eliminated. The same reasoning eliminates terms $\ps{x_2(\bk)}{S_1x_5(\bk)}$ and $\ps{x_3(\bk)}{S_1x_6(\bk)}$. 

\medskip

Several other terms also match: take $U = \begin{pmatrix} I_3 & 0 \\ 0 & I_3 \end{pmatrix}$, then we have $US_1U = -S_1$, and $Ux_1 = x_4$, yielding
\[ \ps{x_1}{S_1x_5} = \ps{x_1}{S_1Ux_2} = -\ps{Ux_1}{S_1x_2} = -\ps{x_4}{S_1x_2}. \]

The overall sum in~\eqref{eq:order2sum} then reduces to: 
\begin{equation*}
    \frac{2 |\ps{x_1(\bk)}{S_1x_5(\bk)}|^2 }{m(\bk)+m(\bk+\mathbf{b_1^*})} 
     + \frac{2 |\ps{x_1(\bk)}{S_1x_6(\bk)}|^2}{m(\bk)+m(\bk+\mathbf{b_2^*})} 
     + \frac{2|\ps{x_2(\bk)}{S_1x_6(\bk)}|^2 }{m(\bk+\mathbf{b_1^*})+m(\bk+\mathbf{b_2^*})}
\end{equation*}
and these scalar products write 
\[ \ps{x_1(\bk)}{S_1x_5(\bk)} = \frac{1}{6}m(\bk+\mathbf{b_1^*})\left( e^{\ri\theta(\bk)}e^{-\ri\theta(\bk + \mathbf{b_2^*})} - e^{\ri\theta(\bk + \mathbf{b_2^*})}e^{-\ri\theta(\bk + \mathbf{b_1^*})}\right),\]		
and so on, which gives us the expressions for $\Theta_i$ in~\eqref{eq:energy_taylor}.

\subsection{Band structure of the Kagome lattice} \label{appendix:kagome}

The Kagome lattice is the line graph of the honeycomb lattice (i.e. that obtained when placing a site on each honeycomb bond), it is also sometimes referred to as the trihexagonal tiling, see Figure~\ref{fig:kagome_lattice}. Consider a finite supercell in the honeycomb lattice, made of $L \times L$ 2-atom cells with periodic boundary conditions, its line graph is a periodic Kagome lattice with $3L^2$ sites (since we count 3 honeycomb bonds per 2-atom cell). Set $A$ to be its adjacency matrix: typically when considering a tight-binding Hamiltonian with hopping strength $t$ on the Kagome lattice, one takes $H = -tA$. 

Our conventions: the Kagome lattice is a triangular Bravais lattice with a three-element basis. The lattice vectors are $\textbf{a}_1, \textbf{a}_2$, and the basis vectors are $0, 1/2\textbf{a}_1, 1/2\textbf{a}_2$. Apply a Fourier transform to obtain $A$'s band structure: if $a_\textbf{R}, b_\textbf{R}, c_\textbf{R}$ are values at site $\textbf{R}$, then define $a_\textbf{k}, b_\textbf{k}, c_\textbf{k}$ for $\bk$ in a discretised Brillouin zone by

\[ \begin{pmatrix}
    a_\textbf{R} \\ b_\textbf{R} \\ c_\textbf{R}
\end{pmatrix}
= \frac{1}{\sqrt{3L^2}}\sum_{\textbf{k}\in B_{2}^{\text{disc}}} \e^{\ri\textbf{k}\cdot \textbf{R}}
\begin{pmatrix}
    a_\textbf{k} \\ \e^{\ri\frac{\textbf{a}_1}{2}\cdot \textbf{k}} \ b_\textbf{k} \\ \e^{\ri\frac{\textbf{a}_2}{2}\cdot \textbf{k}} \ c_\textbf{k}
\end{pmatrix}.\]
In this Fourier basis, $A$ rewrites as 

\[ A(\textbf{k}) = 2\begin{pmatrix}
    0 & \cos\left(\frac{\textbf{a}_1}{2}\cdot \textbf{k}\right) & \cos\left(\frac{\textbf{a}_2}{2}\cdot \textbf{k}\right) \\
    \cos\left(\frac{\textbf{a}_1}{2}\cdot \textbf{k}\right) & 0 & \cos\left(\frac{\textbf{a}_1 - \textbf{a}_2}{2}\cdot \textbf{k}\right) \\
    \cos\left(\frac{\textbf{a}_2}{2}\cdot \textbf{k}\right) & \cos\left(\frac{\textbf{a}_1 - \textbf{a}_2}{2}\cdot \textbf{k}\right) & 0
\end{pmatrix}.\]

Its eigenvalues can be  computed directly, yielding a flat band $\lambda_0 = -2$, and two further bands 

\[\lambda_\pm(\textbf{k}) = 1 \pm \sqrt{3 + 2\cos(\textbf{a}_1\cdot \textbf{k}) + 2\cos(\textbf{a}_2\cdot \textbf{k}) + 2\cos((\textbf{a}_1-\textbf{a}_2)\cdot \textbf{k})   }, \] ranging over $[-2,1]$ and $[1,4]$ respectively for $\textbf{k}$ in the Brillouin zone. Hence the exact lower bound for $A$ is $-2$, it is attained by Kekulé configurations, as well as less symmetric configurations, see e.g. \cite{bergman2008band}.

\begin{figure}
    \centering
    \includegraphics[scale=0.8]{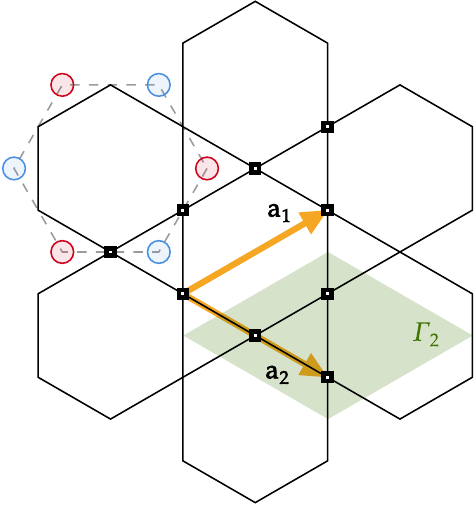}
    \caption{The Kagome lattice (squares and full lines) appears as the line graph of the honeycomb (in dotted lines), with the same lattice vectors $\textbf{a}_1, \textbf{a}_2$ and three sites per $\Gamma_2$ cell. }
    \label{fig:kagome_lattice}
\end{figure}

\subsection{Translation-invariance for large $\mu$, for strongly convex distortion functionals.} \label{appendix:gen_elastic}

We can generalise the proof of the second point in Theorem~\ref{thm:phasetrans} as follows. Take the energy per Carbon atom for some $L$-periodic configuration $\textbf{t}$ to be 

\[
    \mce(\textbf{t}) = -\VTr |T(\textbf{t})| + \frac{\mu}{2}\frac{1}{2L^2}\sum_{i,n} F(t_{i,n})
\]
where $F$ is the elastic energy term. We have proven the result for $F$ is quadratic, and we can generalise simply by assuming $F$ to be strongly convex with a unique minimum at $C > 0$, setting $2\alpha := \min F'' > 0$. Assume also that for negative $t$ that $F(t) \geq F(|t|)$, to ensure $t_* \geq 0$. The pristine energy writes for positive $t$:
\[ \mce(t\textbf{1}) = -t \VTr |H| + \frac{3\mu}{4} F(t) \]
with Euler-Lagrange in $t_*$ writing as

\[ F'(t_*) = \frac{4}{3\mu} \VTr |H|, \]
yielding a unique value for $t_*$ due to the increasing nature of $F'$. Note that $t_* = t_*(\mu)$ is decreasing in $\mu$, and lower bounded by $C$. Expansion of the total distortion term writes as follows, recalling $\textbf{t} = t_* + \bh$ and $\varepsilon = \langle \bh \rangle = \frac{1}{3 L^2} \sum_{i,n} h_{i,n}$:
\begin{align*}
    \frac{\mu}{4L^2}\sum_{i,n} F(t_{i,n}) 
    & = \frac{\mu}{4L^2}\sum_{i,n} F(t_* + h_{i,n}) \\
    & \geq \frac{\mu}{4L^2} \sum_{i,n} F(t_*) + F'(t_*)(h_{i,n}) + \alpha (h_{i,n})^2 \\
    & = \frac{3\mu}{4}F(t_*) + \varepsilon \VTr |H| + \alpha \frac{\mu}{4L^2}\sum_{i,n} (h_{i,n})^2. 
\end{align*}
Now recall that the computation following~\eqref{ineq:en_pre_kag} yielded the quantum kinetic lower bound:
\[ -\underline{\Tr } \ T(\textbf{t}) \geq - t_* \VTr |H| - \varepsilon \VTr |H| + \frac{1}{2t_*} \VTr M \cdot \frac{1}{L^2}\sum_{i,n} (h_{i,n})^2, \]
where $M := -\frac32 \frac{1}{|H|} + \frac16|H|$. Set as previously $\delta^2 = \frac{1}{3L^2} \sum_{i,n} (h_{i,n})^2$, then collecting these two lower bounds yields
\[
\mce(\textbf{t})  \ge \mce(\textbf{t}_*) + \frac{3\delta^2}{2} \left[ \frac{\mu}{2}\alpha - \frac{1}{t_*} \VTr M \right] . 
\]

This again yields the existence of a critical value for $\mu$: since $t_*$ is lower bounded by $C >0$, $\mu t_*$ must diverge with $\mu$, hence there is a value of $\mu$ such that $ \mu\alpha/2~-~\VTr M/t_*$ is positive. Above this critical $\mu$, the pristine configuration will be the only minimizer, with uniqueness ensured by the trace concavity used in the quantum term. 
\\

This argument fails if we take $F$ to be merely convex, for instance by setting $F(t)~=~|t-1|^{3/2}$: if we only apply the trace inequality, we get
\[
\mce(\textbf{t}) \geq \mce(\textbf{t}_*) + \frac{1}{4L^2} \sum_{i,n} \mu\left( |t_* + h_{i,n} -1|^{3/2} - F(t_*)\right) - (h_{i,n})^2\frac{ 1}{t_*} \VTr M .
\]
In the sum, for any fixed value of $\mu$, large $h$ will make the terms negative, thus no threshold $\mu$ can be found in this manner.

%%%%%%%%%%%%%%

\printbibliography

\end{document}